\documentclass[conference, 10pt]{IEEEtran}
\ifCLASSINFOpdf
\else
\fi
\usepackage[cmex10]{amsmath}
\usepackage{amsfonts}
\usepackage{amssymb}
\usepackage{bm}
\usepackage{dsfont}
\usepackage{xcolor}

\newtheorem{theorem}{Theorem}

\newtheorem{remark}{Remark} 

\newcommand{\bra}[1]{\langle #1|}
\newcommand{\ket}[1]{|#1\rangle}
\newcommand{\braket}[2]{\langle #1|#2\rangle}
\DeclareMathOperator{\Tr}{Tr}
\DeclareMathOperator*{\argmin}{arg\,min}

\newcommand{\Espcc}{E_{\text{sp}}^{\text{cc}}}

\newcommand{\Pem}{\mathsf{P}_{\text{e}|m}}
\newcommand{\Pemax}{\mathsf{P}_{\text{e,max}}}

\begin{document}
\title{Constant Compositions in the  Sphere Packing Bound for Classical-Quantum Channels}

\author{\IEEEauthorblockN{Marco Dalai}
\IEEEauthorblockA{Department of Information Engineering\\
University of Brescia, Italy\\
Email: marco.dalai@unibs.it}
\and
\IEEEauthorblockN{Andreas Winter}
\IEEEauthorblockA{ICREA \& F\'isica Te\`orica: Informaci\'o i Fenomens Qu\`antics\\
Universitat Aut\`{o}noma de Barcelona, Spain\\
Email: andreas.winter@uab.cat}
}


%Telephone: (800) 555--1212, Fax: (888) 555--1212}


\maketitle

\begin{abstract}

The sphere packing bound, in the form given by Shannon, Gallager and Berlekamp, was recently extended to classical-quantum channels, and it was shown that this creates a natural setting for combining probabilistic approaches with some combinatorial ones such as the Lov\'asz theta function.
In this paper, we extend the study to the case of constant composition codes. We first extend the sphere packing bound for classical-quantum channels to this case, and we then show that the obtained result is related to a variation of the Lov\'asz theta function studied by Marton. We then propose a further extension to the case of varying channels and codewords with a constant \emph{conditional} composition given a particular sequence. This extension is then applied to auxiliary channels to deduce a bound which can be interpreted as an extension of the Elias bound.
\end{abstract}


% \fi
\IEEEpeerreviewmaketitle

\section{Introduction}
%A number of results in the theory of classical communication through quantum channels have been obtained in the past years that parallel many of the results obtained in the period 1948-1965 for classical channels. ...[BRIEF DISCUSSION?]...  see \cite{holevo-1998}.

%\setlength{\belowdisplayskip}{4pt}
%\setlength{\belowdisplayshortskip}{0pt}
%\newdimen\memlength
%\setlength{\memlength}{\abovedisplayskip}
%\the\abovedisplayskip

%\the\memlength

%\setlength{\abovedisplayskip}{3pt}
%\setlength{\abovedisplayshortskip}{0pt}

%\the\abovedisplayskip

%\setlength{\abovedisplayskip}{\memlength}

%\the\abovedisplayskip

%\begingroup
%\setlength{\abovedisplayskip}{5pt}
%\setlength{\belowdisplayskip}{5pt}

The sphere packing  bound has been recently extended to classical-quantum channels \cite{dalai-ISIT-2012}, \cite[Sec. V]{dalai-TIT-2013} by resorting to the first rigorous proof given for the case of classical discrete memoryless channels (DMC) by Shannon, Gallager and Berlekamp \cite{shannon-gallager-berlekamp-1967-1}.
That resulted in an upper bound to the reliability function of classical-quantum channels, which is the error exponent achievable by means of optimal codes.

The classical proof given in \cite{shannon-gallager-berlekamp-1967-1} can be considered a rigorous completion of Fano's first efforts toward proving the bound \cite[Ch. 9]{fano-book}. However, while Fano's approach led to a tight exponent at high rates for general constant composition codes, the proof in \cite{shannon-gallager-berlekamp-1967-1} only considers the case of the optimal composition.
Shortly afterwards, Haroutunian \cite{haroutunian-1968}, \cite{csiszar-korner-book}, proposed a simple yet rigorous proof which gives the tight exponent for codes with general (possibly non optimal) constant composition. However, a greedy extension of this proof to classical-quantum channels does not give a good bound (see \cite[Th. II.20 and page 35]{winter-phd-1999}). This motivated the choice made in \cite{dalai-ISIT-2012, dalai-TIT-2013} to follow the approach of \cite{shannon-gallager-berlekamp-1967-1}.

In this paper, we modify slightly the approach in \cite{dalai-ISIT-2012, dalai-TIT-2013}  to derive a sphere packing bound for classical-quantum channels with constant composition codes. The main difference with respect to the classical case is in the resulting possible analytical expressions of the bound, which does not seem to be expressible, in this case, in terms of the Kullback-Leibler divegence and mutual information.
In analogy with the results obtained in \cite{dalai-ISIT-2013a} \cite[Sec. VI]{dalai-TIT-2013}, we then discuss the connections of the constant composition version of the bound with a quantity  introduced by Marton \cite{marton-1993} as a generalization of the Lov\'asz theta function  for bounding the zero-error capacity.
Finally, we propose an extension of the sphere packing bound for varying channels and codewords with a constant \emph{conditional} composition from a given sequence, and we show that this result includes as a special case a recently developed generalization of the Elias bound \cite{dalai-ISIT-2014a}.

\section{Definitions}
Consider a classical-quantum channel $\mathcal{C}$ with input alphabet $ \mathcal{X}=\{1,\ldots,| \mathcal{X}|\}$ and associated density operators $S_ x$, $ x\in \mathcal{X}$, in a finite dimensional Hilbert space $\mathcal{H}$. The $n$-fold product channel acts in the tensor product space $\bm{\mathcal{H}}=\mathcal{H}^{\otimes n}$ of $n$ copies of $\mathcal{H}$. To a sequence $\bm{ x}=( x_1, x_2,\ldots, x_n)$ is associated the signal state $\bm{S}_{\bm{ x}}=S_{ x_1}\otimes S_{ x_2}\cdots\otimes S_{ x_n}$.
A block code with $M$ codewords is a mapping from a set of $M$ messages $\{1,\ldots,M\}$ into a set of $M$ codewords  $\bm{ x}_1,\ldots, \bm{ x}_M$, as in the classical case. 
The rate of the code is again $R=(\log M)/n$.

We consider a quantum decision scheme for such a code (POVM) composed of a  collection of $M$ positive operators $\{\Pi_1,\Pi_2,\ldots,\Pi_M\}$ such that $\sum \Pi_m \leq \mathds{1}$, where $\mathds{1}$ is the identity operator. 
The probability that message $m'$ is decoded when message $m$ is transmitted is $\mathsf{P}_{m'|m}=\Tr \Pi_{m'} \bm{S}_{\bm{ x}_m}$ and the probability of error after sending message $m$ is
\begin{equation*}
\Pem=1-\Tr\left(\Pi_m \bm{S}_{\bm{ x}_m}\right).
\end{equation*}
The maximum error probability of the code is defined as the largest $\Pem$, that is,
\begin{equation*}
\Pemax=\max_{m}\Pem.
\end{equation*}
In this paper, we are interested in bounding the probability of error for constant composition codes. Given a composition $P_n$, we define  $\Pemax^{(n)}(R,P_n)$ to be the smallest maximum error probability among all codes of length $n$, rate \emph{at least} $R$, and composition $P_n$.
For a probability distribution $P$, we define the asymptotic optimal error exponent with composition $P$ as
\begin{equation}
E(R,P)=\limsup_{n\to\infty} -\frac{1}{n}\log \Pemax^{(n)}(R,P_n),
\label{eq:E(R)_def_class}
\end{equation} 
where the limsup is over all sequences of codes with rates at least $R$ and compositions $P_n$ tending to $P$ as $n\to\infty$. For channels with a zero-error capacity, the function $E(R,P)$ can be infinite for rates $R$ smaller than some given rate $C_0(P)$, which we can call the zero-error capacity of the channel relative to $P$. It is important to observe that, as for $C_0$, the value $C_0(P)$ only depends on the confusability graph $G$ of the channel, for which we could also call it $C(G,P)$ \cite{csiszar-korner-1981}, \cite{marton-1993}. 

To avoid unnecessary complications, we use a flexible notation in this paper. We keep it simple as far as possible, progressively increasing its complexity by adding arguments to functions as their definitions become more general. The meaning of all quantities will be clear from the context.

\section{Sphere Packing Bound for Constant Composition Codes}

\begin{theorem}
\label{th:cc-sphere-packing}
For all positive rates $R$, distribution $P$, and positive $\varepsilon < R$, we have the bound
\begin{equation*}
E(R,P)\leq \Espcc(R-\varepsilon, P),
\end{equation*}
where $\Espcc(R,P)$ is defined by the relations
\begin{align*}
\Espcc(R,P) & =  \sup_{\rho \geq 0} \left[ E_0^{\text{cc}}(\rho,P) - \rho R\right]%\label{eq:Espcc}
,\\
E_0^{\text{cc}}(\rho,P) & = \min_{F} \left[-(1+\rho)\sum_x P(x) \log\Tr(S_x^{\frac{1}{1+\rho}}F^{\frac{\rho}{1+\rho}})\right].
%\label{eq:E0rhoq}
\end{align*}
the minimum being over all density operators $F$.
\end{theorem}

\begin{remark}
The bound is written here in terms of R\'enyi divergences. For commuting states, that is, classical channels, the bound coincides with the classical one often written in terms of Kullback-Leibler  divergences and mutual information (see \cite[Ch. 5, Prob. 23]{csiszar-korner-book}). That other form is particularly pleasant since it has a very intuitive interpretation. In the case of non-commuting states, however, that interpretation and the associated analytical expression give a bound which is in general weaker than the one given above
 (see \cite[Th. II.20 and page 35]{winter-phd-1999}). It is still interesting to question whether another interpretation could be given in this non-commuting case.
\end{remark}

\begin{proof}
The structure of the proof is the same as in \cite{shannon-gallager-berlekamp-1967-1}, and \cite[Th. 5]{dalai-TIT-2013} and only some technical details must be changed. Due to space limitations, we cannot include here a complete self-contained version; we use the notation adopted in \cite[Th. 5]{dalai-TIT-2013}, recall the basic structure of the proof and  point out what the required changes are.

The idea is again to consider a binary hypothesis test between a properly selected code signal $\bm{S}_{\bm{x}_m}$ and an auxiliary density operator $\bm{F}=F^{\otimes n}$. We can then relate $\Pemax$ and $R$ to the two probabilities of error of the considered test which, using the Chernoff bound, can be bounded (see \cite[Th. 5]{dalai-TIT-2013}) in terms the function $\mu(s)$ defined as
\begin{align}
\mu(s) & =  \log \Tr \bm{S}_{\bm{{x}}_m}^{1-s}\bm{F}^s\\
%&=  \log \prod_{i=1}^n \Tr S_{{x}_i}^{1-s} F^s\\
 %& =  n \sum_{x} P({x})\log \Tr S_{{x}}^{1-s} F^s\\
 & = n \sum_{x} P({x})\mu_{S_x, F}(s) ,
\end{align}
where $\mu_{S_x, F}(s)= \log \Tr S_{{x}}^{1-s} F^s$. 
%we arrive at the result that either
%\begin{equation}
%\Pemax>\frac{1}{8}\exp\left[\mu(s)-s\mu'(s)-s\sqrt{2\mu''(s)}\right]
%\label{eq:cond1}
%\end{equation}
%or
%\begin{equation}
%R<-\frac{1}{n}\left[ \mu(s)+(1-s)\mu'(s)-(1-s)\sqrt{2\mu''(s)} - \log 8\right].
%\label{eq:cond2}
%\end{equation}

%These two equations can then be studied as in \cite[Th. 5]{dalai-TIT-2013}, by choosing an 
An optimal choice of $F$ is then considered for a given $s$, and an optimal $s$ is finally picked depending on the rate $R$. The main difference with respect to \cite[Th. 5]{dalai-TIT-2013} is that here, for a given $s$, the operator $F$ is chosen to be the operator $F_s$ defined as
\begin{equation}
F_s = \argmin_F - \sum_x P(x) \log(\Tr S_x^{1-s}F^s).
\end{equation}
This implies that, instead of bounding $\mu(s)$ using a bound on $\mu_{S_x,F}(s)$ as done in \cite[eqs. (51)-(53)]{dalai-TIT-2013}, we can directly write 
\begin{equation}
\mu(s)=-n(1-s)E_0^{\text{cc}}\left(\frac{s}{1-s},P\right).
\end{equation}
We then proceed essentially as in \cite[Th. 5]{dalai-TIT-2013} by considering sequences of codes of increasing block-length. The main difference here is that the compositions $P_{n_1},P_{n_2},\ldots,P_{n_k},\ldots$ of the codes in the sequence are already known to converge to the given distribution $P$, and the probabilities of error $\Pemax^{(n_1)}, \Pemax^{(n_2)},\ldots,\Pemax^{(n_k)},\ldots$ are such that
\begin{equation*}
E(R,P)=\lim_{k\to\infty} -\frac{1}{n_k}\log \Pemax^{(n_k)}.
\end{equation*}
Then again we proceed as in \cite[Th. 5]{dalai-TIT-2013}, with the only difference that we have now to check the continuity of our new $F_s$ is $s$ in the interval $0<s<1$, which is however even simpler than in that case. Apart from small obvious details, the next difference is in ``case 2)'' of the proof in \cite[Th. 5]{dalai-TIT-2013} where we conclude that 
\begin{align*}
R & \leq  \sum_{x} P({x}) \left(- \frac{1}{s}\mu_{S_{x},F_s}(s)\right).
\end{align*}
In this case we do not need to bound each single term of the sum as done in \cite[eq. (53)]{dalai-TIT-2013};	in this new setting, the right hand side of the above equation is  by definition exactly
\begin{equation}
 \frac{1-s}{s}E_0^{\text{cc}}\left(\frac{s}{1-s}\right)
\end{equation} 
and, from this point onward, we proceed essentially as in \cite[Th.~5]{dalai-TIT-2013}.
\end{proof}

Now it is not difficult to show that after optimization of the composition we recover the original bound of \cite{dalai-ISIT-2012}, \cite{dalai-TIT-2013}. In order to do this, note that 
\begin{align*}
\max_P \Espcc(R) & =  \sup_{\rho \geq 0} \left[ \max_P E_0^{\text{cc}}(\rho,P) - \rho R\right].
\end{align*}
Then, 
\begin{align*}
\max_P  & E_0^{\text{cc}}(\rho,P) \\
&  = \max_P \min_{F} \left[-(1+\rho)\sum_x P(x) \log\Tr(S_x^{\frac{1}{1+\rho}}F^{\frac{\rho}{1+\rho}})\right].\\
& =  \min_{F} \max_P  \left[-(1+\rho)\sum_x P(x) \log\Tr(S_x^{\frac{1}{1+\rho}}F^{\frac{\rho}{1+\rho}})\right]\\
& =  \min_{F} \left[-(1+\rho)\max_x \log\Tr(S_x^{\frac{1}{1+\rho}}F^{\frac{\rho}{1+\rho}})\right],
\end{align*}
where the minimum and the maximum can be exchanged due to linearity in $P$ and convexity in $F$. The resulting expression is in fact the coefficient $E_0(\rho)$ which defines the sphere packing bound as proved in \cite[Th.~6]{dalai-TIT-2013}.

\section{Connections with Marton's function}
The bound $\Espcc(R,P)$ obtained in the previous section can be used as an upper bound for the zero-error capacity of the channel relative to $P$. Whenever the function $\Espcc(R-\varepsilon,P)$ is finite, in fact, then the probability of error at rate $R$ is non-zero. It is not difficult to observe that the smallest rate $R_\infty(P)$ at which $\Espcc(R,P)$ is finite can be evaluated as
\begin{align*}
R_\infty(P) & = \lim_{\rho\to\infty} \frac{E_0^{\text{cc}}(\rho,P)}{\rho}\\
& = \min_{F} \left[-\sum_x P(x) \log\Tr(S_x^0 F)\right],
\end{align*}
where $S_x^0$ is the projection onto the range of $S_x$. When optimized over $P$, we obtain the expression
\begin{align*}
R_\infty & = \min_{F} \max_x \log\frac{1}{\Tr(S_x^0 F)},
\end{align*}
already discussed in \cite{dalai-TIT-2013}. We then have the bounds $C_0(P)\leq R_\infty(P)$ and $C_0\leq R_\infty$.

It was observed in \cite{dalai-ISIT-2013a} and \cite[Sec. VI]{dalai-TIT-2013} that $R_\infty$ is related to the Lov\'asz number $\vartheta$ \cite{lovasz-1979}.
Here, we observe that, in complete analogy, the value $R_\infty(P)$ is related to a variation of the $\vartheta$ function introduced by Marton in \cite{marton-1993} as an upper bound to $C(G,P)$.
Given a (confusability) graph $G$, Marton proposes the following upper bound\footnote{We use the notation $\vartheta(G,P)$ in place of Marton's $\lambda(G,P)$ to preserve a higher coherence with the context of this paper. For the same reason, in what follows we also use, as in \cite{dalai-TIT-2013},  a logarithmic version of the ordinary Lov\'asz $\vartheta$ function, that is, our $\vartheta$ corresponds to $\log\vartheta$ in Lov\'asz' notation. \label{note:loglovasz}} to $C(G,P)$:
\begin{equation}
\vartheta(G,P)
         = \min_{\{u_x\}, f} \sum_x P(x) \log\frac{1}{|\braket{u_x}{f}|^2},
\label{eq:martonstheta}
\end{equation}
where the minimum is over all representations $\{u_x\}$ of the graph $G$ in the Lov\'asz sense and over all unit norm vectors $f$ (in some Hilbert space). Let us now compare this quantity with $R_\infty(P)$. We enforce the notation writing $R_\infty(\{S_x\}, P)$  to point out the dependence of $R_\infty(P)$ on the channel states $S_x$. Now, for a given confusability graph $G$, the best upper bound to $C(G,P)$ is obtained by minimizing $R_\infty(\{S_x\}, P)$ over all possible  channels with confusability graph $G$. We may then define
\begin{align}
\vartheta_{\text{sp}}(G,P)
     & = \inf_{\{S_x\}} R_\infty(\{S_x\}, P)\\
     & = \inf_{\{U_x\}, F} \sum_x P(x) \log\frac{1}{\Tr(U_x F)},
\label{eq:thetasp}
\end{align}
where $\{U_x\}$ now runs over all sets of projectors with confusability graph $G$, and deduce the bound $C(G,P)\leq \vartheta_{\text{sp}}(G,P)$.

The quantity $\vartheta_{\text{sp}}(G,P)$ is the constant composition analog of the formal quantity $\vartheta_{\text{sp}}(G)$ defined in \cite[Sec.~VI]{dalai-TIT-2013}. 
In that case it was observed by Schrijver that in fact $\vartheta_{\text{sp}}(G)=\vartheta(G)$ (with our logarithmic definition of $\vartheta$, see footnote \ref{note:loglovasz}). We have the analogous result for a constant composition.
\begin{theorem}
  For any graph $G$ and composition $P$, $\vartheta_{\text{sp}}(G,P)=\vartheta(G,P)$.
  \label{th:thetaspeqtheta}
\end{theorem}
\begin{proof}
It is obvious that $\vartheta_{\text{sp}}(G,P)\leq\vartheta(G,P)$, since the right hand side of  \eqref{eq:martonstheta} is obtained by restricting the operators in the right hand side of \eqref{eq:thetasp} to have rank one.

We now prove the converse inequality (cf.~\cite{DuanWinter}). 
Let $\{U_x\}$ and $F$ be a representation of $G$ and a state. 
Let first $\ket{\psi}\in \mathcal{H}\otimes \mathcal{H}'$ be a purification of $F$ obtained using an auxiliary space $\mathcal{H}'$, so that $\Tr(U_x F)=\Tr(U_x\otimes \mathds{1}_{\mathcal{H}'}\ket{\psi}\bra{\psi})$. Let then 
\begin{equation}
\ket{w_x}=\frac{U_x\otimes \mathds{1}_{\mathcal{H}'}\ket{\psi}}{\|U_x\otimes \mathds{1}_{\mathcal{H}'}\ket{\psi}\|}.
\end{equation}
It is not difficult to check that $\{w_x\}$ is an orthonormal representation of $G$ and that 
$\Tr(U_x F) = \Tr(U_x\otimes \mathds{1}_{\mathcal{H}'}\ket{\psi}\bra{\psi})
            = |\braket{w_x}{\psi}|^2$, for all $x$.
Hence, the orthormal representation $\{w_x\}$ and the unit norm vector $\psi$ satisfy
\begin{equation}
  \sum_x P(x) \log\frac{1}{\Tr(U_x F)}
         = \sum_x P(x)\log\frac{1}{|\braket{w_x}{\psi}|^2},
\end{equation}
which implies that $\vartheta(G,P) \leq \vartheta_{\text{sp}}(G,P)$.
\end{proof}

We can now discuss another interesting issue about the use of the quantity $\vartheta(G,P)$. When we are interested in bounding $C_0$, we can use the bound $C_0\leq \vartheta(G)$ or we can also use the bound\footnote{Not that $C_0 = \max_P C_0(P)$, since the number of compositions is polynomial in the block-length.}  $C_0\leq \max_P \vartheta(G,P)$.
Marton \cite{marton-1993} stated that this does not make a difference, since $\max_P \vartheta(G,P)=\vartheta(G)$. However, a proof of this statement does not seem to follow easily from the definitions, since we would need to exchange the maximization over $P$ with the minimization over representations and handles.
We use Theorem \ref{th:thetaspeqtheta} to prove this statement.
\begin{theorem}
  For any graph $G$, $\max_P \vartheta(G,P)=\vartheta(G)$.
\end{theorem}
\begin{proof}
For any representation $\{U_x\}$ of $G$ and density operator $F$, define
the function $f(x) = \Tr U_x F$, and denote the set of all functions $f$
obtained in this way by $\text{OR}(G)$.
The proof of Theorem \ref{th:thetaspeqtheta} shows that any $f\in\text{OR}(G)$
can be realized by rank-one projections $U_x = \ket{u_x}\bra{u_x}$
and a pure state $F = \ket{f}\bra{f}$, in a space of dimension at most
$|\mathcal{X}|$ (namely the span of the $\ket{u_x}$). 
In particular, it follows that $\text{OR}(G)$ is closed and compact.

Furthermore, it is convex: namely, consider $f_i(x) = \Tr U_x^{(i)} F^{(i)}$
for representations $\{U_x^{(i)}\}$ of $G$ and density operators $F^{(i)}$,
$i=1,2$. Then, for $0\leq p \leq 1$, let
$U_x = U_x^{(1)} \oplus U_x^{(2)}$ and 
$F = p F^{(1)} \oplus (1-p) F^{(2)}$, which has associated
$f(x) = \Tr U_x F = p f_1(x) + (1-p) f_2(x)$, i.e.~$p f_1 + (1-p) f_2 \in \text{OR}(G)$.

Now define the quantity
\begin{equation}
  J(f,P)=\sum_x P(x)\log \frac{1}{f(x)},
\end{equation}
for compositions $P$ and functions $f\in\text{OR}(G)$.
The theorem is equivalent to the statement that
\begin{equation}
  \max_P \min_{f\in\text{OR}(G)} J(f,P)= \min_{f\in\text{OR}(G)} \max_P J(f,P),
  \label{eq:minimax}
\end{equation}
since the left hand side equals $\max_P \vartheta(G,P)$ by 
Theorem \ref{th:thetaspeqtheta}, 
and the right hand side equals $\vartheta(G)$ by \cite[Th. 8]{dalai-TIT-2013}.

But \eqref{eq:minimax} is an instance of the minimax theorem. Indeed, both
the domains of $f$ and $P$ are convex and compact, and the functional
$J$ is convex in the former and concave (in fact affine linear) in
the latter.
\end{proof}

We close this section with a simple yet useful result which we will need in the next section. This is the analogous of \cite[Th. 10]{dalai-TIT-2013} for the constant composition setting.
\begin{theorem}
For any pure-state channel we have the inequality $\Espcc(R_\infty(P),P)\leq R_\infty(P)$.
\label{th:Erinfty}
\end{theorem}
\begin{proof}
For a pure state channel, since $S_x^{\frac{1}{1+\rho}}=S_x$, we have
\begin{align*}
E_0^{\text{cc}}(\rho,P) & = \min_{F} \left[-(1+\rho)\sum_x P(x) \log\Tr(S_x F^{\frac{\rho}{1+\rho}})\right]\\
& \leq  \min_{F} \left[-(1+\rho)\sum_x P(x) \log\Tr(S_x F)\right]\\
& = (1+\rho) R_\infty(P),
\end{align*}
from which we easily deduce the statement by definition of $\Espcc(R,P)$.
\end{proof}

\section{A Conditional Sphere Packing Bound}
We now propose an extension of the sphere packing to handle the case of varying channels
with a \emph{conditional composition} constraint on the codewords.
%Although this setting  can be considered artificial, the bound will prove useful when applied to auxiliary channels in a procedure that can be considered as the evolution of the method used in \cite[Sec. VIII]{dalai-TIT-2013} along the same lines used in \cite{dalai-ISIT-2014a}.
Here we assume that we have a finite set $\mathcal{A}$ of possible states and a different channel $\mathcal{C}_a$, for each state $a\in\mathcal{A}$. The communication is governed by a sequence of states $\bm{a}=(a_1,\ldots,a_n)\in\mathcal{A}^n$ (known to both encoder and decoder)  with composition $P$, which determines the channels to use. In particular, channel $\mathcal{C}_{a_i}$ is used at time instant $i$. The composition constraint in this case is that all codewords have conditional composition $V$ given $\bm{a}$, which means that any codeword has a symbol $x$ in a fraction $V(x|a)$ of the $nP(a)$ positions where $a_i=a$.
Note that this general scenario includes the ordinary constant composition situation described before, which is obtained for example when $P(a)=1$ for some $a$ and $\bm{a}=(a,a,\ldots,a)$. For a given $P$ and $V$, let now $E(\{\mathcal{C}_a\},R,P,V)$ be the optimal asymptotic error exponent achievable by codes with asymptotic conditional composition $V$ with respect to a sequence with asymptotic composition $P$ using the set of channels $\{\mathcal{C}_a\}$, $a\in\mathcal{A}$.

We can adapt the proof of the sphere packing bound by choosing the density operator $\bm{F}$ to take into account this state dependent structure of the communication process. In particular, instead of using $n$ identical copies of a single density operators $F$, we can use $|\mathcal{A}|$ different operators $F_a$, $a\in\mathcal{A}$ to build $\bm{F}$ as
\begin{equation}
\bm{F}=F_{a_1}\otimes F_{a_2} \otimes \cdots \otimes F_{a_n}.
\end{equation}
The theorem can be extended to this case without substantially changing the proof, the main difference being in the function $\mu(s)$ which now reads
\begin{equation}
\mu_{\bm{S}_{\bm{{x}}_m},\bm{F}}(s) =  n \sum_{a, x} P({a})V(x|a)\mu_{S_{x},F_a}(s).
\end{equation}
This leads to a bound in the form
\begin{equation}
E(\{\mathcal{C}_a\},R,P,V)\leq \Espcc(\{\mathcal{C}_a\},R-\varepsilon, P,V),
\end{equation}
where $\Espcc(\{\mathcal{C}_a\},R,P,V)$ is defined by
\begin{align}
\Espcc(\{\mathcal{C}_a\},R,P,V) & =  \sup_{\rho \geq 0} \left[ E_0^{\text{cc}}(\{\mathcal{C}_a\},\rho,P,V) - \rho R\right],\\
E_0^{\text{cc}}(\{\mathcal{C}_a\},\rho,P,V) & = \sum_a P(a) E_0^{\text{cc}}(\mathcal{C}_a,\rho,V(\cdot|a)).
\end{align}
and $E_0^{\text{cc}}(\mathcal{C}_a,\rho,V(\cdot|a))$ is the coefficient $E_0^{\text{cc}}$ of the sphere packing bound for channel $\mathcal{C}_a$ with composition $V(\cdot|a)$.

This bound is finite for all rates $R\geq R_\infty(\{{\mathcal{C}}_a\},P,V)$ where 
\begin{equation}
R_\infty(\{{\mathcal{C}}_a\},P,V) = \sum_a P(a) R_\infty({\mathcal{C}}_a,V(\cdot|a)),
\end{equation}
and it is not difficult to show, using the same procedure used in Theorem \ref{th:Erinfty}, that for pure state channels we have the inequality 
\begin{equation}
\Espcc(\{\mathcal{C}_a\},R_\infty(\{{\mathcal{C}}_a\},P,V),P,V)\leq R_\infty(\{{\mathcal{C}}_a\},P,V).
\label{eq:CondEspRinfty}
\end{equation}

We can now combine this bound with the ideas presented in \cite{dalai-ISIT-2013b}, \cite{dalai-TIT-2013} and \cite{blahut-1977}, much in the same way as done in \cite{dalai-ISIT-2014a}, to obtain a bound on the reliability of a channel $\mathcal{C}$ using auxiliary classical-quantum channels $\{\tilde{\mathcal{C}}_a\}$.  We limit here the discussion to the case of a pure-state channel with states $S_x=\ket{\psi_x}\bra{\psi_x}$ and pure-states auxiliary channels  $\{\tilde{\mathcal{C}}_a\}$.
For a $\rho\geq 1$, we define the set $\Gamma(\rho)$ of admissible pure-state auxiliary channels $\tilde{\mathcal{C}}$ with states $\tilde{S}_x=\ket{\tilde{\psi}_{x}}\bra{\tilde{\psi}_{x}}$ such that
\begin{equation}
|\braket{\tilde{\psi}_{x}}{\tilde{\psi}_{x'}}|\leq |\braket{\psi_{x}}{\psi_{x'}}|^{1/\rho} , \quad \forall x,x' \in \mathcal{X}.
\end{equation}
For any $a\in\mathcal{A}$ we choose an auxiliary pure state channel $\tilde{\mathcal{C}}_a\in\Gamma(\rho)$ with states $\tilde{S}_{a,x}=\ket{\tilde{\psi}_{a,x}}\bra{\tilde{\psi}_{a,x}}$. Given a sequence $\bm{a}=(a_1,\ldots,a_n)\in\mathcal{A}^n$ and a sequence $\bm{x}=(x_1\ldots,x_n)\in\mathcal{X}^n$, let 
\begin{equation}
\tilde{\bm{\psi}}_{\bm{a},\bm{x}} = \tilde{\psi}_{a_1,x_1}\otimes\cdots\otimes \tilde{\psi}_{a_n,x_n}.
\end{equation}
Now, given two sequences $\bm{x}=(x_1,\ldots,x_n)$ and $\bm{x}'=(x'_1,\ldots,x'_n)$, we can use these auxiliary channels to bound the overlap $ |\braket{\bm{\psi}_{\bm{x}}}{\bm{\psi}_{\bm{x}'}}|^2$  as
\begin{equation}
 |\braket{\bm{\psi}_{\bm{x}}}{\bm{\psi}_{\bm{x}'}}|^2 \geq |\braket{\tilde{\bm{\psi}}_{\bm{a},\bm{x}}}{\tilde{\bm{\psi}}_{\bm{a},\bm{x}'}}|^{2\rho}.
\end{equation}
This will allow us to bound $E(R,P)$ for the original channel using the bound (see for example \cite[Th. 12]{dalai-TIT-2013})
\begin{align}
E(R,P) & \leq -\frac{1}{n}\log \max_{m\neq m'} |\braket{\bm{\psi}_{\bm{x}_m}}{\bm{\psi}_{\bm{x}_{m'}}}|^2 + o(1)
\label{eq:ineqEinprod}\\
& \leq -\frac{\rho}{n}\log \max_{m\neq m'}|\braket{\tilde{\bm{\psi}}_{\bm{a},\bm{x}_m}}{\tilde{\bm{\psi}}_{\bm{a},\bm{x}_{m'}}}|^2 + o(1).
\label{eq:ineq1}
\end{align}
We could use the extension of the sphere packing bound considered in this section to upper bound the right hand side of the last equation as done in \cite[Sec. VIII]{dalai-TIT-2013} if all codewords $\bm{x}_m$ had the same conditional composition given the sequence $\bm{a}$. Since the sequence $\bm{a}$ is arbitrary, we choose it so that this condition is met by at least a large enough subset $\mathcal{T}$ of codewords, and we only apply the sphere packing bound to this subset $\mathcal{T}$. In order to do this, we adopt an idea proposed by Blahut \cite{blahut-1977} in a generalization of the Elias bound and already considered for a further generalization in \cite{dalai-ISIT-2014a}.

Given a code with $M=e^{nR}$ codewords of composition $P$, assume that there exist conditional compositions $V(x'|x)$ (i.e., $nP(x)V(x'|x)$ is an integer) such that
\begin{equation}
\sum_{x}P(x)V(x'|x)=P(x')
\end{equation}
(that we will write as $PV=P$) and
\begin{equation}
R> I(P,V),
\end{equation}
where $I(P,V)$ is the mutual information with the notation of \cite{csiszar-korner-book}.
Then, (see \cite{blahut-1977}, proof of Th. 8) there is at least one sequence $\bar{\bm{x}}$ of composition $P$ (not necessarily a codeword) such that there is a subset $\mathcal{T}$ of at least $|\mathcal{T}|=e^{n(R- I(P,V)-o(1))}$ codewords with conditional composition $V$ from $\bar{\bm{x}}$. 
We now choose the set $\mathcal{A}=\mathcal{X}$ and the sequence $\bm{a}=\bar{\bm{x}}$, although we keep the original notation $\bm{a}$ to avoid confusion. Furthermore, since we are interested in the limit as $n\to\infty$, we directly work with the asymptotic composition $P$ and matrix $V$, removing the constraint that $nP(x)$ and $nP(x)V(x'|x)$ are integers.

Now, we can use the conditional sphere packing bound introduced here to bound the probability of error of the subcode $\mathcal{T}$ used over the varying channel $\tilde{C}_{a_1},\cdots,\tilde{C}_{a_n}$.  For these codewords used over this varying channel, there is a decision rule such that (\cite{holevo-2000}, \cite[Sec. VIII]{dalai-TIT-2013})
\begin{align}
\tilde{\mathsf{P}}_{\text{e,max}} & \leq (|\mathcal{T}|-1)\max_{m,m'\in \mathcal{T}}|\braket{\tilde{\bm{\psi}}_{\bm{a},\bm{x}_m}}{\tilde{\bm{\psi}}_{\bm{a},\bm{x}_{m'}}}|^2\\
& \leq e^{n(R-I(P,V)+o(1))}\max_{m,m'\in \mathcal{T}}|\braket{\tilde{\bm{\psi}}_{\bm{a},\bm{x}_m}}{\tilde{\bm{\psi}}_{\bm{a},\bm{x}_{m'}}}|^2.
\label{eq:ineq2}
\end{align}
On the other hand, as $n\to\infty$ we have 
\begin{multline}
-\frac{1}{n}\log \tilde{\mathsf{P}}_{\text{e,max}}\leq \Espcc(\{\tilde{\mathcal{C}}_a\},R-I(P,V)-\varepsilon, P,V) \\+ R - I(P,V).
\label{eq:ineq3}
\end{multline}
Putting together equations \eqref{eq:ineq1}, \eqref{eq:ineq2} and \eqref{eq:ineq3}, we obtain
\begin{multline}
E(R,P)\leq \rho [\Espcc(\{\tilde{\mathcal{C}}_a\},R-I(P,V)-\varepsilon, P,V) \\ +R-I(P,V)].
\end{multline}
Since the choice of $\rho$, of the channels $\{\tilde{\mathcal{C}}_a\}\in\Gamma(\rho)$ and of $V$ can be optimized, we have, in analogy with \cite[Th. 11]{dalai-TIT-2013},
\begin{theorem}
The reliability function with constant composition $P$ satisfies $E(R,P)\leq E_{\text{spu}}^{\text{cc}}(R,P)$ where
\begin{multline}
E_{\text{spu}}^{\text{cc}}(R,P)=\inf \rho [\Espcc(\{\tilde{\mathcal{C}}_a\},R-I(P,V)-\varepsilon, P,V) \\ +R-I(P,V)],
\label{def:Espucc}
\end{multline}
the infimum being over $\varepsilon>0$, $\rho\geq 1$, auxiliary channels $\tilde{\mathcal{C}}_a\in\Gamma(\rho)$, and conditional distributions $V$ such that $PV=P$.
\label{Th:Espucc}
\end{theorem}
\begin{remark}
Note that for the choice $V(x'|x)=P(x')$, $\forall x$, we have $I(P,V)=0$. We can also notice that the optimization of the channels $\tilde{\mathcal{C}}_a$ will give $\tilde{\mathcal{C}}_a=\tilde{\mathcal{C}}$, $\forall a$, for an optimal $\tilde{\mathcal{C}}$. With this constraint on $V$, the bound $E_{\text{spu}}^{\text{cc}}(R,P)$ is weakened to
\begin{equation}
\inf \rho [\Espcc(\tilde{\mathcal{C}},R-\varepsilon, P) +R],
\end{equation}
where the minimum is now only over $\rho\geq 1$ and $\tilde{\mathcal{C}}\in \Gamma(\rho)$. This is a constant composition version of the bound in \cite[Th. 11]{dalai-TIT-2013}.
\end{remark}

In the same way as \cite[Th. 11]{dalai-TIT-2013} generalizes the results of \cite[Sec. III]{dalai-TIT-2013}, Theorem \ref{Th:Espucc} generalizes the results of \cite{dalai-ISIT-2014a}.	To see this, we can study the smallest rate for which the bound $E_{\text{spu}}^{\text{cc}}(R,P)$ is finite. First note that for fixed channels $\{\tilde{\mathcal{C}}_a\}$, distribution $V$, and $\varepsilon$ sufficiently small, the quantity on the right hand side of equation \eqref{def:Espucc} is finite for $R>R_\infty(\{\tilde{\mathcal{C}}_a\},P,V)+I(P,V)$. Furthermore, when $R$ approaches this value from the right,  using equation \eqref{eq:CondEspRinfty}, 
the right hand side of equation \eqref{def:Espucc} is upper bounded by $2\rho (R_\infty(\{\tilde{\mathcal{C}}_a\},P,V))$. Using now equation \eqref{eq:ineqEinprod},
we find that 
\begin{equation*}
 -\frac{1}{n}\log \max_{m\neq m'} |\braket{\bm{\psi}_{\bm{x}_m}}{\bm{\psi}_{\bm{x}_{m'}}}|^2 \geq  2\rho (R_\infty(\{\tilde{\mathcal{C}}_a\},P,V)) + o(1).
\end{equation*}
So, for $R>R_\infty(\{\tilde{\mathcal{C}}_a\},P,V)$ we have the bound
\begin{equation*}
 -\frac{1}{n}\log \max_{m\neq m'} |\braket{\bm{\psi}_{\bm{x}_m}}{\bm{\psi}_{\bm{x}_{m'}}}| \geq  \rho (R_\infty(\{\tilde{\mathcal{C}}_a\},P,V)) + o(1).
\end{equation*}
Optimizing now over $\rho$, $V$ and the auxiliary channels $\{\tilde{\mathcal{C}}_a\}$, and comparing the definition of $R_\infty(\{\tilde{\mathcal{C}}_a\},P,V)$ with the definition of $\vartheta(\rho,P,V)$ used in \cite{dalai-ISIT-2014a}, we find that the	 bound of Theorem \ref{Th:Espucc} includes, as a particular case, the bound presented in \cite{dalai-ISIT-2014a} as a generalization of the Elias bound on the Bhattacharyya distance of codes.


%\section*{Acknowledgment}




%\bibliographystyle{IEEEtran}
%\begin{thebibliography}{1}

%\bibliography{bibeit}



%\end{thebibliography}

\end{document}